%% file: Receiver_Memroy.tex
\documentclass[conference, twcolumn]{IEEEtran}
\usepackage{graphicx,cite,calc,dsfont}
\usepackage{hyperref}
\usepackage{amsmath}
\usepackage{amsthm,amssymb}
\usepackage{algorithm}
\usepackage{algpseudocode}
\newtheorem{thm}{Theorem}
\newtheorem{lem}{Lemma}

\theoremstyle{definition}
\newtheorem{examp}{Example}

\usepackage{empheq}
\usepackage{breqn}
\usepackage{comment}


\include{defs}

\begin{document}
\title{Cloud-Aided Interference Management with Cache-Enabled Edge Nodes and Users}
\author{ Seyed Pooya Shariatpanahi$^{1,2}$, Jingjing Zhang$^{3}$, Osvaldo Simeone$^{3}$,\\ Babak Hossein Khalaj$^{4}$, Mohammad-Ali Maddah-Ali$^{5}$\\[4mm] 

1. School of Electrical and Computer Engineering, University of Tehran, Tehran, Iran \\
2. School of Computer Science, Institute for Research in Fundamental Sciences (IPM), Tehran, Iran  \\
3. Department of Informatics, King's College London, London, UK \\
4. Department of Electrical Engineering, Sharif University of Technology, Tehran, Iran\\
5. Nokia Bell Labs, Holmdel, NJ, USA \\
\vspace{-0.9cm}
\thanks{ACK}}

\maketitle

\begin{abstract}
This paper considers a cloud-RAN architecture with cache-enabled multi-antenna Edge Nodes (ENs) that deliver content to cache-enabled end-users. The ENs are connected to a central server via limited-capacity fronthaul links, and, based on the information received from the central server and the cached contents, they transmit on the shared wireless medium to satisfy users' requests. By leveraging cooperative transmission as enabled by ENs' caches and fronthaul links, as well as multicasting opportunities provided by users' caches, a close-to-optimal caching and delivery scheme is proposed. As a result, the minimum Normalized Delivery Time (NDT), a high-SNR measure of delivery latency, is characterized to within a multiplicative constant gap of $3/2$ under the assumption of uncoded caching and fronthaul transmission, and of one-shot linear precoding. This result demonstrates the interplay among fronthaul links capacity, ENs' caches, and end-users' caches in minimizing the content delivery time.
\end{abstract}

\section{Introduction}

Caching content at the network edge can mitigate the heavy traffic burden at network peak times. Contents are proactively stored in caches at the Edge Nodes (ENs) or at the end-users during low-traffic periods, relieving network congestion at peak hours \cite{Kangasharju-2002,Nygren-2010}. Edge caching at the ENs can enable cooperative wireles transmission in the presence of shared cached contents across multiple ENs \cite{Liu-2015, Sengupta-2017, Naderializadeh-2017}. In contrast, caching of shared content at the users  enables the multicasting of coded information that is useful simultaneously for multiple users \cite{MaddahAli-2014,Karamchandani-2016,Pedarsani-2016,Zhang-2017-0}.


In practice, not all contents can be cached, and requested uncached contents should be fetched from a central server through finite-capacity fronthaul links. This more general set-up, illustrated in Fig.~1, was studied in \cite{Sengupta-2017, Tandon-2016,Azimi-2017,Zhang-2017} (see also reference therein) in the absence of users' caches. These references consider as the performance metric of interest the overall delivery latency, including both fronthaul and wireless contributions. In particular, in prior works \cite{Naderializadeh-2017}\cite{Madda-15}, the delivery latency is measured in the high Signal-to-Noise Ratio (SNR) regime. While \cite{Sengupta-2017, Tandon-2016,Azimi-2017} allow any form of delivery strategy, including interference alignment, in \cite{Zhang-2017}, the optimal high-SNR latency performance is studied under the assumption that wireless transmission can only use practical one-shot linear precoding strategies. Reference \cite{Zhang-2017} presents a caching-fronthaul-wireless transmission scheme that is shown to be latency-optimal within a multiplicative factor of $3/2$.


\begin{figure}[t!] 
  \centering
\includegraphics[width=0.65\columnwidth]{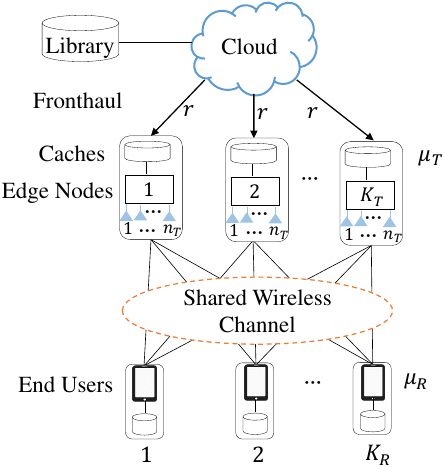}
\caption{Cloud-RAN system with cache-enabled ENs and end-users.} 
\label{fig:model}
\vspace{-1.7em}
\end{figure}

In this paper, we extend the results in \cite{Zhang-2017} by allowing caching not only at the ENs but also at the end-users. To this end, we consider the cloud-RAN scenario in Fig.~1 and evaluate the impact of both cooperative transmission opportunities at the ENs and multicasting opportunities brought by caching at the users. Caching at the users has, in fact, two potentially beneficial effects on the network performance. First, since users have already cached some parts of the library, they need not receive from the network the cached portion of the requested file --- this is known as the \emph{local caching gain}. Second, as assumed, by a careful cache content placement, a common coded message can benefit more than one user, which is known as the \emph{global caching gain} \cite{MaddahAli-2014}. Assuming that entire library is cached across ENs and users and that the fronthaul links are absent, reference \cite{Naderializadeh-2017} proved that the gains accrued form cooperative transmission by the ENs and the global caching gain provided by users' caches are additive. Here, we generalize this conclusion by considering the role of finite-capacity fronthaul links and by allowing for partial caching of the library of popular files across ENs and users. 


The rest of the paper is organized as follows. In Section \ref{Sec-SystemModel} we describe the system model. Section \ref{Sec-MainResult} presents the main results along with an intuitive discussion. In Section \ref{Sec-Achievable} we detail the proposed caching and delivery scheme. Then, we derive a converse in Section \ref{Sec-Converse}, which is proved to be within a multiplicative gap of 3/2 as compared to the high-SNR performance achievable by the proposed scheme. Finally, Section \ref{Sec-Conclusion} concludes the paper.

\section{System Model}\label{Sec-SystemModel}
We consider a content delivery scenario, illustrated in Fig.~1, in which $K_T$ Edge Nodes (ENs), each with $n_T$ antennas, deliver requested contents to $K_R$ single-antenna users via a shared wireless medium. The contents library includes $N$ files, each of $L$ bits, which are collected in set $\mathcal{W}=\{W_1,\ldots,W_N\}$. Furthermore, each file $W_n$ is divided into $F$ packets, collected in the set $\mathcal{W}_n=\{W_{nf}\}_{f=1}^{F}$, where $F$ is an arbitrary integer and each packet consists of $L/F$ bits. Each EN is connected to a central server, where the library resides, via a wired fronhaul link of capacity $C_F$ bits per symbol of the wireless channel. Moreover, each EN is equipped with a cache of size $\mu_T N$ files, for $\mu_T\leq 1$. In this paper, in contrast to \cite{Zhang-2017}, we assume that the users are also cache-enabled, each with a cache of size $\mu_R N$ files, for $\mu_R \leq 1$. Henceforth, for simplicity, we assume that both $\mu_T K_T$ and $\mu_R K_R$ are integers, with extensions following directly as in \cite{Sengupta-2017,Zhang-2017, Naderializadeh-2017}

The system operation includes two phases, namely the \emph{cache content placement} and \emph{content delivery} phases. In the first phase, each EN and each user caches \emph{uncoded} fractions of the files in the library at network off-peak traffic hours and without knowing the actual requests of the users in the next phase. In the second phase, at network peak traffic hours, at any transmission slot, each active user requests access to one of the files in the library, i.e., user $k \in \{1,\ldots,K_R\}$ requests file $W_{d_k}$, $d_k \in \{1,\dots,N\}$. For delivery, first, the cloud sends on each fronthaul link some \emph{uncoded} fractions of the requested files to the ENs. For more general ways to use the fronthaul links, we refer to \cite{Sengupta-2017}. After fronthaul transmission, the ENs collaboratively deliver the requested contents to the users via the edge wireless downlink channel based on the cached contents and fronthaul signals. 

The signal received by each user $k$ on the downlink channel is given as
\begin{equation} \label{trans}
    y_k=\sum_{i=1}^{K_T} \mathbf{h}_{ki}^H \mathbf{x}_i + z_k,
\end{equation}
in which $\mathbf{h}_{ki} \in \mathbb{C}^{n_T \times 1}$ is the complex representation of the fading channel vector from EN $i$ to user $k$; $\mathbf{x}_i \in \mathbb{C}^{n_T \times 1}$ is the transmitted vector from EN $i$; $z_k$ is unit-power additive Gaussian noise; and $(.)^H$ represents the Hermitian transpose. The fading channels are drawn from a continuous distribution and are constant in each transmission slot. The transmission power of each EN is constrained by $\mathbb{E} \left[||\mathbf{x}_i||^2\right] \leq \mathrm{SNR}$. Furthermore, as in \cite{Zhang-2017,Naderializadeh-2017}, the ENs transmit using one-shot linear precoding, so that the vector transmitted by each EN at time slot $t$ is given as 
\begin{equation}
    \mathbf{x}_i = \sum_{(n,f)} \mathbf{v}_{inf} s_{nf},
\end{equation}
where $s_{nf}$ is a symbol encoding file fraction $W_{nf}$, and $\mathbf{v}_{inf}$ is the corresponding beamforming vector. Furthermore, we assume that Channel State Information (CSI) is available to all the entities in the network. 

The performance metric of interest is the Normalized Delivery Time (NDT) introduced in \cite{Sengupta-2017}, which measures the high-SNR latency due to fronthaul and wireless transmissions.
To this end, we write $C_F=r\log\text{SNR}$, hence allowing the fronthaul capacity to scale with the wireless edge $\mathrm{SNR}$ with a scaling constant $r\geq 0$. Then, denoting the time required to complete the fronthaul and wireless edge transmissions as $T_F$ and $T_E$ (measured in symbol periods of the wireless channel) respectively, the total NDT is defined as the following limit over SNR and file size 
\begin{align} \label{def:NDT}
\delta&=\delta_F+\delta_E \\ \nonumber
&=\lim_{\text{SNR} \rightarrow \infty} \lim_{L \rightarrow \infty} \left(\frac{E[T_F]}{L / \log(\mathrm{SNR})} +\frac{E[T_E]}{L / \log(\mathrm{SNR})} \right).
\end{align}
In \eqref{def:NDT}, the term $L / \log(\mathrm{SNR})$ represents the normalizing delivery time on an interference-free channel; the term $\delta_F=\lim_{\text{SNR} \rightarrow \infty} \lim_{L \rightarrow \infty} E[T_F]/(L / \log(\text{SNR}))$ is defined as the fronthaul NDT; and $\delta_E=\lim_{\text{SNR} \rightarrow \infty} \lim_{L \rightarrow \infty} E[T_E]/(L / \log(\text{SNR}))$ as the edge NDT. 

Accordingly, for given cloud and caching resources defined by the triple $(r,\mu_T,\mu_R)$, the minimal NDT over all achievable policies is defined as
\begin{equation} \label{def:min}
\delta^*(r,\mu_T,\mu_R)=\inf\{\delta(r,\mu_T,\mu_R): \delta(r,\mu_T,\mu_R) \, \textrm{is achievable}\},
\end{equation}
where the infimum is over all uncoded caching, uncoded fronthaul, and one-shot linear edge transmissions policies that ensure reliable delivery for any set of requested files \cite{Zhang-2017,Sengupta-2017}. 



\section{Main Result}\label{Sec-MainResult}

In this section we state our main result and its implications. We proceed by first proposing an achievable scheme and then proving its optimality within a constant multiplicative gap. 

In the cache content placement phase, the scheme follows the standard approach of sharing a distinct fraction of a file to all subsets of $\mu_T K_T$ ENs and $\mu_R K_R$ users, hence satisfying the cache capacity constraints \cite{Naderializadeh-2017}. As a result, each fraction of any requested file is available at $m_R=\mu_R K_R$ users, which we define as \emph{receive-side} multiplicity, and at $\mu_T K_T$ ENs. As we will see, in the content delivery phase, the \emph{transmit-side multiplicity} $m_T$, i.e., the number of ENs at which any fraction of a requested files is available, can be increased beyond $\mu_T K_T$ by means of fronthaul transmission. 

As proved in \cite{Naderializadeh-2017}, and briefly reviewed  below, the content multiplicities $m_T$ and $m_R$ can be leveraged in order to derive a delivery scheme that serves simultaneously \begin{align} \label{user}
   u(m_T,m_R) =\min(K_R, n_T m_T+m_R) 
\end{align}
 users at the maximum high-SNR rate of $\log$(SNR). Unlike \cite{Naderializadeh-2017}, however, here the transmit-side multiplicity $m_T$ is not fixed, since any uncached fraction of a file can be delivered to an EN by the cloud on the fronthaul. The multiplicity $m_T$ can be hence increased at the cost of a larger fronthaul delay $\delta_F$. Therefore, the multiplicity $m_T$ should be chosen carefully, by accounting for the fronthaul latency $\delta_F$ as well as for the wireless NDT $\delta_E$, which decreases with the size of the number $u(m_T,m_R)$ of users that can be served simultaneously. Our main result below obtains an approximately optimal solution in terms of minimum NDT. 

Before detailing the main result, we briefly present how the scheme in \cite{Naderializadeh-2017} serves $u(m_T,m_R)$ users simultaneously at rate $\log$(SNR) by leveraging both multicasting and cooperative Zero-Forcing (ZF) precoding. Assume that $n_T m_T+m_R\leq K_R$ (the complement case follows in a similar way). At any given time, $m_T+m_R$ ENs transmit simultaneously to deliver fractions of the requested files to $n_T m_T+m_R$ users. To this end, the active ENs are grouped into all subsets of $m_T$ active ENs. Note that there are $\binom{m_T+m_R}{m_T}$ such groups, and that each EN generally belongs to multiple groups. All groups transmit at the same time, with each group delivering collaboratively a shared fraction of a file to the requesting user. Transmission by a group is done within the null space of the channel of other $n_T m_T-1$ active users by means of Zero-Forcing (ZF) one-shot linear precoding. The interference created by this transmission to the remaining $m_R$ active users is removed by leveraging the information in the receive-side caches. This is possible since the caching strategy ensures that the message transmitted by a group of $m_T$ ENs is also available to $m_R$ users. Note that the scheme in \cite{Naderializadeh-2017} assumes $n_T=1,$ but the extension described above is straightforward.

Based on the above mentioned achievable scheme, along with an optimized transmit-side multiplicity, the following theorem characterizes the minimum NDT \eqref{def:min} to within a multiplicative constant equal to $3/2$.

\begin{thm}[Multiplicative gap on minimum NDT]\label{Th-MainTheorem}
 The NDT
\begin{align}\label{Eq_Optimum_NDT}
 \delta_{up}(r,\mu_T,\mu_R)&=\frac{K_R \left(m(r,\mu_T,\mu_R)- \mu_T K_T\right)^+}{K_T r}  \notag \\
 &+ \frac{K_R(1-\mu_R)}{\min \left\{K_R , n_T m(r,\mu_T,\mu_R) +K_R \mu_R\right\}}
\end{align}
is achievable, where 
\begin{equation}\label{Eq-Optimum-Multiplicity}
    m(r,\mu_T,\mu_R)=
    \begin{cases}
      m(r,\mu_R) & \text{if} ~\mu_T K_T < m(r,\mu_R), \\
      \mu_T K_T  & \text{if} ~ m(r,\mu_R) \leq \mu_T K_T \leq m_{max}, \\
      m_{max} & \text{if}~ \mu_T K_T > m_{max},
    \end{cases}
\end{equation}
and
\vspace{-0.5cm}
\begin{equation}\label{Eq_Optimum_Mr} 
~~~~~~m(r,\mu_R)\!=\!\!
    \begin{cases}
     \left[\sqrt{\frac{K_T(1-\mu_R)r}{n_T}}-\frac{K_R \mu_R}{n_T}\right] &\text{if} ~ r<r_{th}, \\
     m_{max} &\text{if} ~ r \geq r_{th},
    \end{cases}
\end{equation}
with \vspace{-0.5cm}
\begin{equation}
r_{th}=\frac{n_T}{K_T(1-\mu_R)} \left(m_{max}+\frac{K_R \mu_R}{n_T}\right)^2,
\end{equation}
and
\begin{equation} \label{mmax}
m_{max} = \min \left\{K_T, \left\lceil{ \frac{K_R(1-\mu_R)}{n_T}} \right\rceil  \right\}.
\end{equation}
Moreover, the minimum NDT satisfies the inequalities
\begin{equation} \label{lower}
   \frac{2}{3}  \delta_{up}(r,\mu_T,\mu_R) \leq \delta^*(r,\mu_T,\mu_R) \leq  \delta_{up}(r,\mu_T,\mu_R).
\end{equation}
\end{thm}



Theorem \ref{Th-MainTheorem} implies that the multiplicity $m_T=m(r,\mu_R,\mu_T)$ in \eqref{Eq-Optimum-Multiplicity} is optimal in terms of NDT, up to a constant multiplicative gap.  Importantly, in contrast to \cite{Zhang-2017}, the choice of $m_T$ in \eqref{Eq-Optimum-Multiplicity} depends also on the caching capacity $\mu_R$ at the users, and it reduces to selection in  \cite{Zhang-2017} when $\mu_R=0$. The first term in \eqref{Eq_Optimum_NDT} is the fronthaul NDT $\delta_F$ required to convey the uncached portions of files to achieve the desired multiplicity $m_T=m(r,\mu_T,\mu_R)$ in \eqref{Eq-Optimum-Multiplicity}. The second term is the edge transmission NDT $\delta_E$, which accounts for the local caching gain (i.e., $(1-\mu_R)$), and for the combined global caching gain due to the users' caches and for the cooperation gain due to the ENs' caches and to fronthaul transmission (i.e., $n_T m + K_R \mu_R$). The result  hence generalizes the main conclusion from  \cite{Shariatpanahi-2016} and \cite{Naderializadeh-2017} that the gains from coded caching multicasting opportunities at the receive side and cooperation at the transmit side are additive.  

\begin{figure}[t!] 
  \centering
\includegraphics[width=0.8\columnwidth]{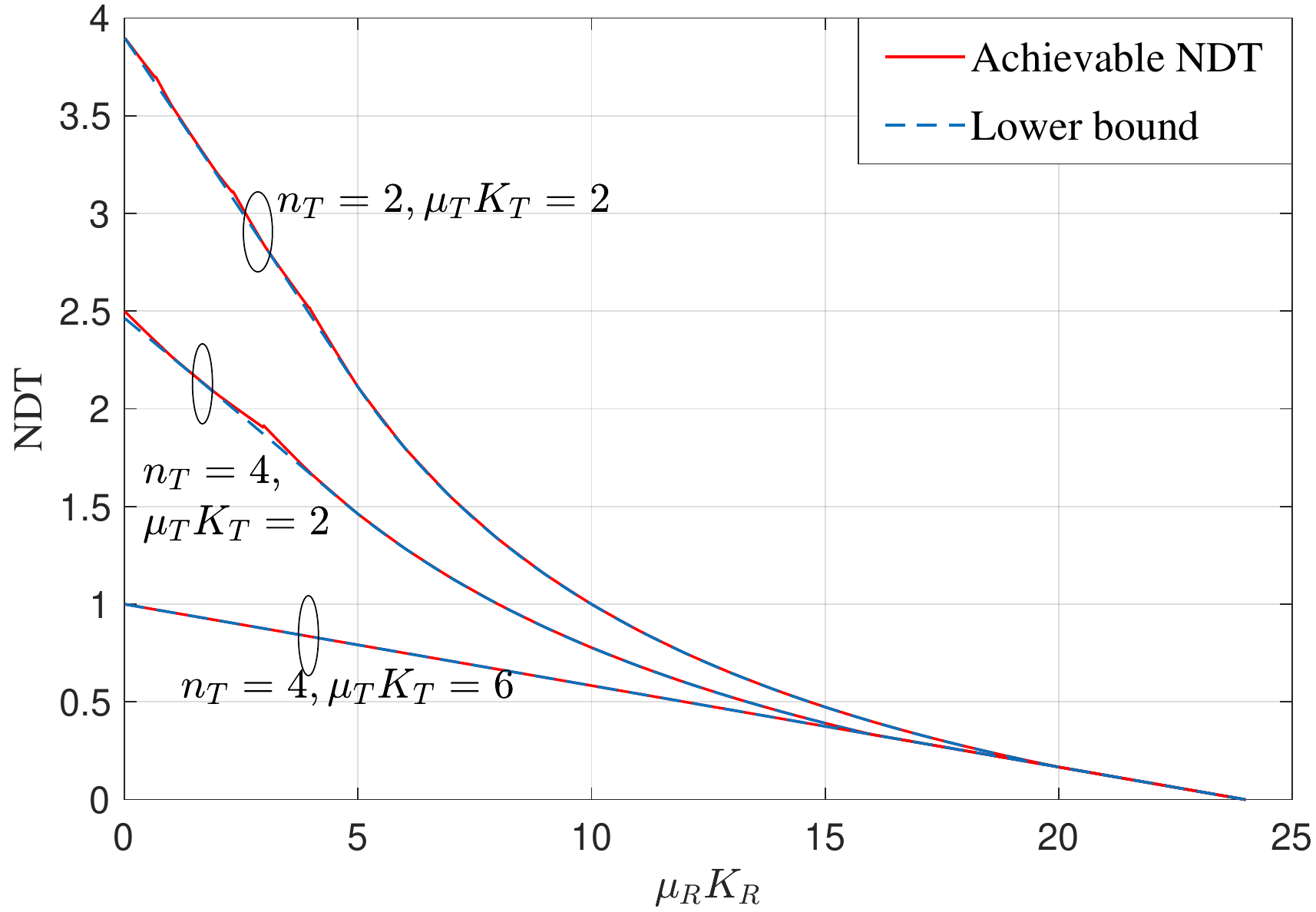}
\vspace{-0.3cm}
\caption{Achievable NDT $\delta_{up}(r,\mu_R,\mu_T)$ and lower bound $\delta^*(r,\mu_R,\mu_T)$ versus $\mu_T K_T$ for different values of $\mu_T$ and $n_T$, with $K_T = 12$, $K_R = 24$  and $r=4$.} 
\label{fig:NDT}
\vspace{-0.6cm}
\end{figure}

\begin{examp}
The achievable NDT $\delta_{up}(r,\mu_T,\mu_R)$ in \eqref{Eq_Optimum_NDT}, along with the lower bound $\delta^*(r,\mu_T,\mu_R)$ derived in Lemma~\ref{lem:conv:optimization} in Section \ref{Sec-Converse}, are plotted in Fig.~\ref{fig:NDT} as a function of the users’ cache capacity $\mu_R$ for different values of the parameters $n_T$, $r$ and $\mu_T$. We set the number of ENs and users to $K_T = 12$ and $K_R = 24$ and the fronthaul rate to $r=4$. Note that for non-integer values of $\mu_R K_R$, the achievable NDT is obtained by memory sharing between the receive-side multiplicities $\lfloor \mu_R K_R\rfloor$ and $\lceil \mu_R K_R\rceil$ \cite{integer}. It is observed that caching at the end-users is more effective when the number of EN transmit antennas and/or the transmit-side caches are small. Furthermore, when the transmit-side multiplicity is sufficient to serve all $K_R$ users at the same time, end-user caching only provides local caching gains. In particular, this happens when $\mu_T K_T=6$ and $n_T=4$, in which case the NDT is seen to decrease linearly with $\mu_R K_R$.
\end{examp}


\section{Achievable Scheme}\label{Sec-Achievable}
The achievable scheme generalizes the strategies proposed in \cite{Zhang-2017} and \cite{Naderializadeh-2017} by accounting for fronthaul transmission and for the caches available at the users. As discussed in Section~\ref{Sec-MainResult}, the cache content placement phase uses the same approach proposed in \cite{Naderializadeh-2017}, which guarantees content replication of $\mu_T K_T$ and $m_R=\mu_R K_R$ at the transmit and receive sides, respectively. 

In the content delivery phase, fronthaul transmission provides packets from the requested files to the ENs in order to increase the transmit-side multiplicity $m_T$ to the desired value $m$. This is at the cost of the fronthaul delay 
\begin{equation}\label{eq:FronthaulNDT}
\delta _F(m) = \frac{K_R \left(m-\mu_T K_T\right)^+}{K_T r},
\end{equation}
given that $(m-\mu_T K_T)^+L$  bits need to be delivered for each requested file (see also \cite{Zhang-2017}).

Based on the multiplicities $m_R$ and $m_T$, the number of users that can be served at the same time is \eqref{user}. Since each user has cached a $(1-\mu_R)$-fraction of its requested file, the edge NDT is given by \cite{Zhang-2017}
\begin{equation} \label{eq:edgeNDT}
\delta_E(m)=\frac{K_R(1-\mu_R)}{u(m)}.
\end{equation}

The transmit-side multiplicity $m_T=m$ should be tuned such that the total delivery latency $\delta_E(m)+\delta _F(m)$ is minimized. First, we determine the maximum possible multiplicity $m_{max}$  from the following necessary conditions
\begin{align}
    m \leq K_T , \quad     n_Tm+ K_R \mu_R \leq K_R,
\end{align}
which result in $m_{max}$ given in \eqref{mmax}.
To proceed, we first focus on the case of $\mu_T=0$, and find a close-to-optimal multiplicity $m(r,\mu_R)$. Then, based on the expression for $m(r,\mu_R)$, we propose a specific choice for the multiplicity for the general case where $\mu_T\geq0$. 

To start, in the case when $\mu_T=0$, from \eqref{eq:FronthaulNDT} and \eqref{eq:edgeNDT}, the total NDT is
\begin{eqnarray} \label{Eq:subsecB-mult}
\delta (m) = \frac{K_R m}{K_T r} + \frac{K_R(1-\mu_R)}{u(m)}.
\end{eqnarray}
In order to optimize over $m$, we find the (only) stationary point for function~\eqref{Eq:subsecB-mult} as
\begin{equation} \label{m_0} 
m_0=\sqrt{\frac{K_T(1-\mu_R)r}{n_T}}-\frac{K_R \mu_R}{n_T}.
\end{equation}
We then approximate the integer solution of the original problem to be the nearest positive integer smaller than $m_{max}$, yielding \eqref{Eq_Optimum_Mr}.

For the general case $\mu_T \geq 0$, we propose the choice \eqref{Eq-Optimum-Multiplicity} for the transmit-side multiplicity. Accordingly, when $\mu_T K_T < m(r,\mu_R)$, and hence the transmit-side caches are small, packets are sent over the fronthaul links so that the aggregate multiplicity is equal to the value $m(r,\mu_R)$ selected above when $\mu_T=0$. For the case $\mu_T\geq m(r,\mu_R)$, instead, the transmit-side multiplicity \eqref{Eq-Optimum-Multiplicity} only relies on EN caching, and fronthaul transmission is not carried out. In particular, when  $\mu_T\geq m_{max}$, the maximum multiplicity $m_{max}$ can be guaranteed directly by EN caching. Theorem 1 demonstrates the near-optimality of this choice. 

As illustration for how the user cache capacity $\mu_R K_R$ affects transmit-side multiplicity $ m(r,\mu_T,\mu_R)$ is shown in Fig.~\ref{fig:m} for $K_T = 12$, $K_R = 24$, $n_T=\mu_T K_T=4$. As $\mu_R K_R$ increases, user-side caching becomes more effective, and less EN-side cooperation is needed to null out interference. Accordingly, the transmit-side multiplicity decreases with $\mu_R K_R$, and it depends on $r$ only when $\mu_R K_R$ is sufficiently small. 

\begin{figure}[t!] 
  \centering
\includegraphics[width=0.7\columnwidth]{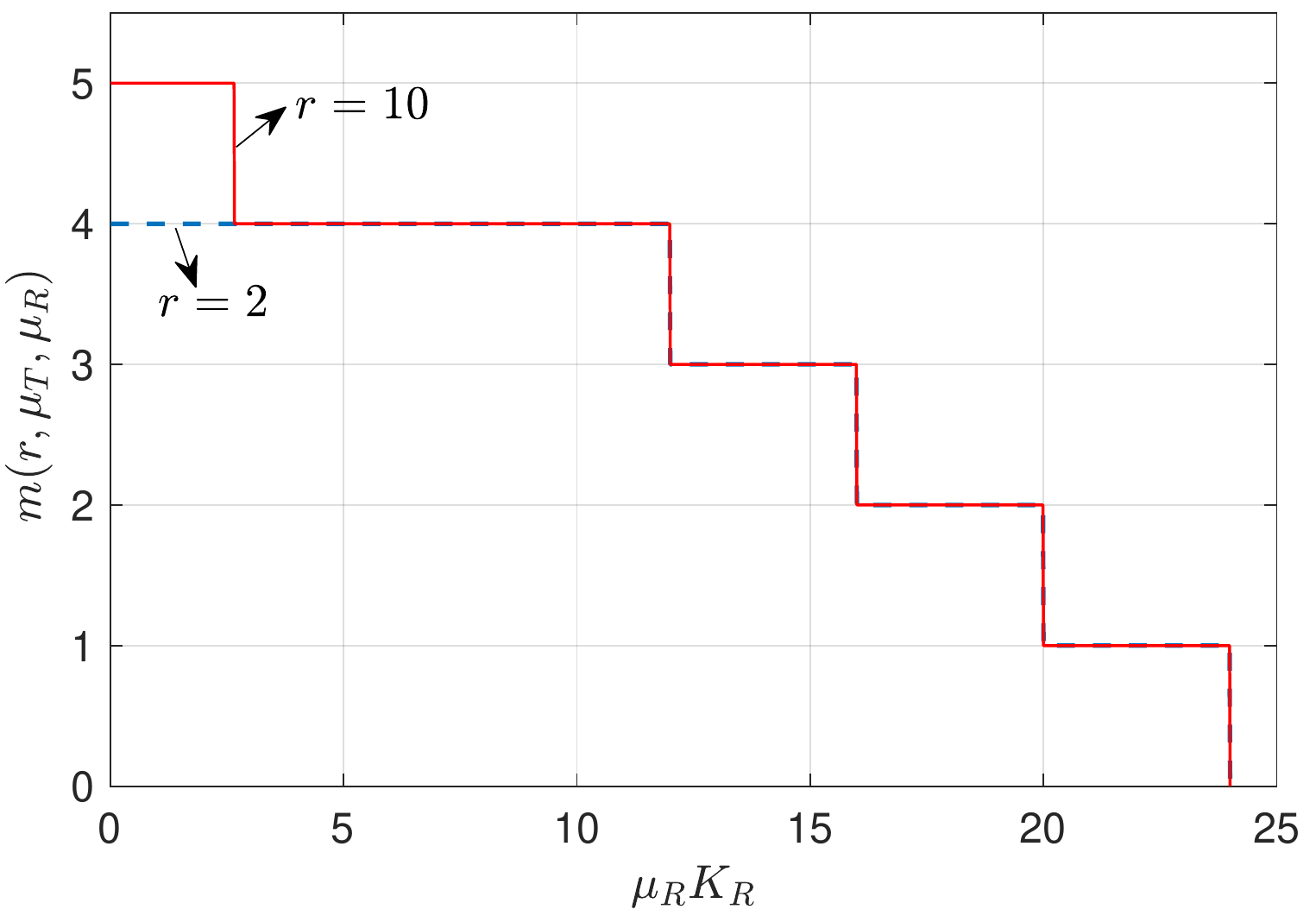}
\vspace{-0.3cm}
\caption{Transmit-side multiplicity $m(r,\mu_T,\mu_R)$ in \eqref{Eq-Optimum-Multiplicity} versus $K_R \mu_R$ for $K_T = 12$, $K_R = 24$, $n_T=\mu_T K_T=4$ and different values of $r$.} 
\label{fig:m}
\vspace{-0.5cm}
\end{figure}

\section{Multiplicative Optimality}\label{Sec-Converse}
In this section we demonstrate that the achievable NDT in Theorem \ref{Th-MainTheorem} is within a multiplicative constant gap of the minimum NDT by proving \eqref{lower}. To this end, we extend the converse proof developed in \cite{Zhang-2017} in order to account for the presence of users' caches. First, without loss of generality, consider a split of each file $W_n$ into $(2^{K_T}-1)\times2^{K_R}$ subsets of packets $\mathcal{W}_n=\{\mathcal{W}_{n\tau_T\tau_R}\}$, such that each part $\mathcal{W}_{n\tau_T\tau_R}$ is indexed by the subsets of indices $\tau_T \subseteq [K_T] \backslash \{\emptyset \}$ and $\tau_R \subseteq [K_R]$. The subset $\mathcal{W}_{n\tau_T\tau_R}$ includes the packets of $\mathcal{W}_n$ that are present at all the ENs $i \in \tau_T$ after fronthaul transmissions and at all the users $j \in \tau_R$. We also define $c_{n\tau_T\tau_R}$ as the number of packets of file $W_n$ that are cached at all the ENs in $\tau_T$, and at all the users in $\tau_R$; $f_{n\tau_T\tau_R}(\mathbf{d})$ to be number of packets from file $W_n$ that are transmitted to all the users in $\tau_T$ via the fronthaul for a given demand vector $\mathbf{d}$ and cached at all the users in $\tau_R$. Note that these quantities are well-defined for every policy. With these definitions, NDT of any achievable policy can be lower bounded by the solution to the following optimization problem:
\begin{align} \label{Eq_MainOpt_Converse}
    &(a) \min_{\substack{\{c_{n\tau_T\tau_R}\} \\ \{f_{n\tau_T\tau_R}(\mathbf{d})\}}} \max_{\mathbf{d}} \quad \delta_E^*\left(\{c_{n\tau_T\tau_R}\}, \{f_{n\tau_T\tau_R}(\mathbf{d})\}, \mathbf{d}\right)+\delta_F^*\left(\mathbf{d}\right)
 \\[-3pt] \nonumber
    & \; \mathrm{subject \quad to}\nonumber  \\ \nonumber
    &(b) \; \sum_{i=1}^{K_T} \sum_{\substack{\tau_T \subseteq [K_T] \\ |\tau_T|=i}} \sum_{j=0}^{K_R} \sum_{\substack{\tau_R \subseteq [K_R] \\ |\tau_R|=j}} \left(c_{n\tau_T\tau_R}+f_{n\tau_T\tau_R}(\mathbf{d})\right) = F,  \\[-10pt] \nonumber  
    &\hspace{2.3in} \forall n \in \mathbf{d}, \forall \mathbf{d}  \\ \nonumber   
    &(c) \; \sum_{n=1}^{N} \sum_{\substack{\tau_T \subseteq [K_T] \\ i \in \tau_T}} \sum_{j=0}^{K_R} \sum_{\substack{\tau_R \subseteq [K_R] \\ |\tau_R|=j}} c_{n\tau_T\tau_R} \leq \mu_T F N, \quad \forall i \in [K_T]  
\end{align}
\begin{align}
    &(d) \; \sum_{n=1}^{N} \sum_{\substack{\tau_R \subseteq [K_R] \\ j \in \tau_R}} \sum_{i=1}^{K_T} \sum_{\substack{\tau_T \subseteq [K_T] \\ |\tau_T|=i}}  \left(c_{n\tau_T\tau_R}+f_{n\tau_T\tau_R}(\mathbf{d})\right) \leq \mu_R F N,  \nonumber \\[-10pt]
    & \hspace{2.3in}\forall j \in [K_R], \forall \mathbf{d}  \nonumber  \\
    &(e) \; \frac{1}{Fr} \sum_{n \in \mathbf{d}} \sum_{\substack{\tau_T \subseteq [K_T] \\ i \in \tau_T}} \sum_{j=0}^{K_R} \sum_{\substack{\tau_R \subseteq [K_R] \\ |\tau_R|=j}} f_{n\tau_T\tau_R}(\mathbf{d}) \leq \delta_F^*(\mathbf{d}), \nonumber\\[-10pt]
    & \hspace{2.3in}\forall i \in [K_T], \forall \mathbf{d}  \nonumber  \\
    &(f)  \{c_{n\tau_T\tau_R},f_{n\tau_T\tau_R}(\mathbf{d})\} \geq 0,  ~~ 0\leq \delta_F^*(\mathbf{d}) \leq \delta_{F_{max}}, 
\end{align}
where function $\delta_E^*(\{c_{n\tau_T\tau_R}\}, \{f_{n\tau_T\tau_R}(\mathbf{d})\}, \mathbf{d})$ is implicitly defined as the minimum edge NDT in \eqref{def:NDT} for given cache and fronthaul policies when the request vector is $\mathbf{d}$, while $\delta_F^*(\mathbf{d})$ is a function for $\mathbf{d}$ that satisfies conditions (e) and (f). We have also defined 
\begin{align}
    \delta_{F_{max}} \triangleq \frac{K_R(m_{max}-\mu_T K_T)^+}{K_T r}.
\end{align}
In \eqref{Eq_MainOpt_Converse}, the equality (b) guarantees the availability of all the requested files; inequalities (c) are due to the fact that the size of the cached content of each EN $i\in[K_T]$ is limited by the cache capacity $\mu_T FN$; similarly, inequalities (d) enforce the cache capacity constraint at each user $j\in[K_R]$;  inequalities (e) follow from the definition of fronthaul NDT in \eqref{def:NDT}, since the left-hand side is the number of packets sent to EN $i$ via the fronthaul link; in (f), $\delta_F^*(\mathbf{d})$ is upper bounded by $\delta_{F_{max}}$ since the maximum multiplicity is $m_{max}$ and hence the total number of bits required via fronthaul is $K_R(m_{max}-\mu_T K_T)^+$. 

In \eqref{Eq_MainOpt_Converse}, the expression of function $\delta_E^*(\{c_{n\tau_T\tau_R}\}, \{f_{n\tau_T\tau_R}(\mathbf{d})\}, \mathbf{d})$ is generally unknown. Notwithstanding this complication, the following Lemma gives a lower bound to the solution of the above optimization problem. Proof can be found in Appendix~\ref{pro:opt}.
\begin{lem}\label{lem:conv:optimization}
    The minimum value of the optimization problem in \eqref{Eq_MainOpt_Converse} is lower bounded by 
    \begin{align}\label{Eq-Lem1}
     f(x)=\frac{K_R}{K_Tr}\left(x - \mu_T K_T\right) + \frac{K_R(1-\mu_R)}{n_T x + \mu_R K_R},
    \end{align}
    where we have defined
    \begin{align} \nonumber
        &x \triangleq \frac{1}{NF} \sum_{i=1}^{K_T} \sum_{j=0}^{K_R} i b_{ij}, ~~~        \tilde{f}_{n\tau_T\tau_R} \triangleq \frac{\sum_{\mathbf{d}: n \in \mathbf{d}} f_{n\tau_T\tau_R}(\mathbf{d}) }{K_R\pi(N-1,K_R-1)}, \\ 
    &\mathrm{and} \quad b_{ij} \triangleq \sum_{\substack{\tau_T \subseteq [K_T] \\ |\tau_T|=i}}  \sum_{\substack{\tau_R \subseteq [K_R] \\ |\tau_R|=j}} \sum_{n=1}^{N} (c_{n\tau_T\tau_R}+\tilde{f}_{n\tau_T\tau_R}).
    \end{align}
\end{lem}

Finally, the following lemma analyzes the gap between the lower bound derived in Lemma \ref{lem:conv:optimization} and the upper bound \eqref{Eq_Optimum_NDT}, completing the proof of Theorem \ref{Th-MainTheorem}. A proof can be found in Appendix~\ref{pro:gap}.
\begin{lem} \label{lem:gap}
   Function $f(x)$ in \eqref{Eq-Lem1} is lower bounded by $\frac{2}{3}\delta_{up}(r,\mu_T,\mu_R)$, where $\delta_{up}(r,\mu_T,\mu_R)$ is given in Theorem \ref{Th-MainTheorem}.
\end{lem}


\section{Conclusions}\label{Sec-Conclusion}
For a cache-enabled cloud-RAN architecture where both the ENs and the end-users have caches, this paper has characterized the minimum delivery latency in the high-SNR to within a multiplicative gap of $3/2$. Under the practical constraint that the ENs can only transmit using one-shot linear precoding, the main result shows that the cooperation gains accrued by EN cooperation via EN caching and fronthaul transmission are additive with respect to the multicasting gains offered by end-user caching.

\section*{Acknowledgements}
Jingjing Zhang and Osvaldo Simeone have received funding from the European Research Council (ERC) under the European Union's Horizon 2020 Research and Innovation Programme (Grant Agreement No. 725731).

\bibliographystyle{ieeetr}


\newpage
\begin{appendices}
\section{Proof of Lemma \ref{lem:conv:optimization}}\label{pro:opt}
To prove Lemma~\ref{lem:conv:optimization}, we substitute the maximum over all the possible request vectors in the objective of \eqref{Eq_MainOpt_Converse} with an average over them. The solution of the resulting problem yields a lower bound to the solution of the orginal problem. Mathematically, the objective (a) in \eqref{Eq_MainOpt_Converse} is substituted with
\begin{align}  \label{min}
     \min_{\substack{\{c_{n\tau_T\tau_R}\} \\ \{f_{n\tau_T\tau_R}(\mathbf{d})\}}} &\frac{1}{\pi(N,K_R)} \\ \nonumber
     & \times \sum_{\mathbf{d}} \left[ \delta_E^*\left(\{c_{n\tau_T\tau_R}\}, \{f_{n\tau_T\tau_R}(\mathbf{d})\}, \mathbf{d}\right)+\delta_F^*\left(\mathbf{d}\right) \right],
\end{align}
where we have defined $\pi(N,K_R)\triangleq N! / (N-K_R)!$. In order to deal with the unknown function $\delta_E^*(\{c_{n\tau_T\tau_R}\},\{f_{n\tau_T\tau_R}(\mathbf{d})\}, \mathbf{d})$, we need the following lemma.
\begin{lem}\label{lem-MuxCodedCaching-Gain}
Define $\tau_{Tl}$ as the subset of edge nodes that have access to the packet $W_{n_lf_l}$ after the fronhaul transmission, and $\tau_{Rl}$ as the subset of users that have cached the packet $W_{n_lf_l}$. Then, the number $u$ of users that can be served at the same time is upper bounded by
\begin{equation}
    u \leq \min_{l \in [u]} |\tau_{Tl}| n_T + |\tau_{Rl}|.
\end{equation}
\end{lem}
\begin{proof}
    This lemma is generalization of Lemma 1 in \cite{Zhang-2017} to the case where we have cache at the users. The main proof is considering the role of receivers' caches in a similar way as in Lemma 3 of \cite{Naderializadeh-2017}, and the rest of the proof is the same. 
\end{proof}


Using the above lemma we will have the following lower bound on the minimum edge NDT:
\begin{align}
    &\delta_E^*\left(\{c_{n\tau_T\tau_R}\}, \{f_{n\tau_T\tau_R}(\mathbf{d})\}, \mathbf{d}\right) \geq \frac{(1-\mu_R)}{F} \\ \nonumber
    & \times \sum_{k=1}^{K_R} \sum_{i=1}^{K_T} \sum_{\substack{\tau_T \subseteq [K_T] \\ |\tau_T|=i}} \sum_{j=0}^{K_R} \sum_{\substack{\tau_R \subseteq [K_R] \\ |\tau_R|=j}} \frac{c_{d_k\tau_T\tau_R}+f_{d_k\tau_T\tau_R}(\mathbf{d})}{i n_T+j},
\end{align}
since at most $i n_T+j$ users can be served simultaneously when the multiplicities at the ENs and the users are $i$ and $j$, respectively. 
Now we lower bound the first term in \eqref{min} as
\begin{align} \nonumber
    &\frac{1}{\pi(N,K_R)} \sum_{\mathbf{d}} \delta_E^*\left(\{c_{n\tau_T\tau_R}\}, \{f_{n\tau_T\tau_R}(\mathbf{d})\}, \mathbf{d}\right) \stackrel{(a)}  \geq  \frac{1-\mu_R}{F\pi(N,K_R)} \\ \nonumber
    &\times \sum_{\mathbf{d}}\sum_{k=1}^{K_R} \sum_{i=1}^{K_T} \sum_{\substack{\tau_T \subseteq [K_T] \\ |\tau_T|=i}} \sum_{j=0}^{K_R} \sum_{\substack{\tau_R \subseteq [K_R] \\ |\tau_R|=j}} \frac{c_{d_k\tau_T\tau_R}+f_{d_k\tau_T\tau_R}(\mathbf{d})}{i n_T+j} \\ \nonumber
    &=\frac{1-\mu_R}{F\pi(N,K_R)} \sum_{i=1}^{K_T} \sum_{j=0}^{K_R} \frac{1}{i n_T+j} \times \\ \nonumber
    & \quad \quad \quad \sum_{\mathbf{d}} \sum_{k=1}^{K_R} \sum_{\substack{\tau_T \subseteq [K_T] \\ |\tau_T|=i}}  \sum_{\substack{\tau_R \subseteq [K_R] \\ |\tau_R|=j}} (c_{d_k\tau_T\tau_R}+f_{d_k\tau_T\tau_R}(\mathbf{d})) \\ \nonumber
    &\stackrel{(b)}=\frac{1-\mu_R}{F\pi(N,K_R)} \sum_{i=1}^{K_T} \sum_{j=0}^{K_R} \frac{K_R \pi(N-1,K_R-1)}{i n_T+j} \times \\ \nonumber
    & \quad \quad  \quad \sum_{\substack{\tau_T \subseteq [K_T] \\ |\tau_T|=i}}  \sum_{\substack{\tau_R \subseteq [K_R] \\ |\tau_R|=j}} \sum_{n=1}^{N} (c_{n\tau_T\tau_R}+\tilde{f}_{n\tau_T\tau_R}) \\ \nonumber
    &=\frac{K_R(1-\mu_R)}{NF} \sum_{i=1}^{K_T} \sum_{j=0}^{K_R} \frac{1}{i n_T+j} \times \\ \nonumber
    & \quad \quad  \quad \sum_{\substack{\tau_T \subseteq [K_T] \\ |\tau_T|=i}}  \sum_{\substack{\tau_R \subseteq [K_R] \\ |\tau_R|=j}} \sum_{n=1}^{N} (c_{n\tau_T\tau_R}+\tilde{f}_{n\tau_T\tau_R}) \\ \nonumber
    &\stackrel{(c)}=\frac{K_R(1-\mu_R)}{NF} \sum_{i=1}^{K_T} \sum_{j=0}^{K_R} \frac{b_{ij}}{i n_T+j} \\ \nonumber
    &\stackrel{(d)} \geq \frac{K_R(1-\mu_R)}{NF}  \frac{\left(\sum_{i=1}^{K_T} \sum_{j=0}^{K_R}b_{ij}\right)^2}{\sum_{i=1}^{K_T} \sum_{j=0}^{K_R}(i n_T+j)b_{ij}} \\ \nonumber
    &\stackrel{(e)} = \frac{K_R(1-\mu_R)}{NF}  \frac{(NF)^2}{\sum_{i=1}^{K_T} \sum_{j=0}^{K_R}(i n_T+j)b_{ij}} \\ \nonumber
    & = \frac{K_R(1-\mu_R)}{\frac{n_T}{NF}\sum_{i=1}^{K_T} \sum_{j=0}^{K_R}ib_{ij}+\frac{1}{NF}\sum_{i=1}^{K_T} \sum_{j=0}^{K_R}jb_{ij}},
\end{align}
where (a) holds because of Lemma~\ref{lem-MuxCodedCaching-Gain}; in (b) we have defined
\begin{align}
    \tilde{f}_{n\tau_T\tau_R} \triangleq \frac{\sum_{\mathbf{d}: n \in \mathbf{d}} f_{n\tau_T\tau_R}(\mathbf{d}) }{K_R\pi(N-1,K_R-1)};
\end{align}
in (c) we have defined
\begin{align}
    b_{ij} \triangleq \sum_{\substack{\tau_T \subseteq [K_T] \\ |\tau_T|=i}}  \sum_{\substack{\tau_R \subseteq [K_R] \\ |\tau_R|=j}} \sum_{n=1}^{N} (c_{n\tau_T\tau_R}+\tilde{f}_{n\tau_T\tau_R});
\end{align}
in (d) we have used the inequality
\begin{align}
    \left(\sum_{i,j}{u_{ij}v_{ij}}\right) ^ 2 \leq \left(\sum_{i,j}{u_{ij}^2}\right)\left(\sum_{i,j}{v_{ij}^2}\right);
\end{align}
with  $u_{ij}=\sqrt{b_{ij} / (in_T+j)}$ and $v_{ij}=\sqrt{b_{ij} \times (in_T+j)}$; and finally (e) results from the equality  $\sum_{i=1}^{K_T} \sum_{j=0}^{K_R}b_{ij}=NF$. This results from summing up (18b) for all $\pi(N,K_R)$ request vectors and for all $K_R$  files in each request vector $\mathbf{d}$ as follows:
\begin{align} \label{eq:bij}
    &\pi(N,K_R) K_R F \\ \nonumber
    &= \sum_{\mathbf{d}} \sum_{n \in \mathbf{d}} \sum_{i=1}^{K_T} \sum_{\substack{\tau_T \subseteq [K_T] \\ |\tau_T|=i} }  \sum_{j=0}^{K_R} \sum_{\substack{\tau_R \subseteq [K_R] \\ |\tau_R|=j} } (c_{n \tau_T \tau_R}+f_{n \tau_T \tau_R}(\mathbf{d})) \\ \nonumber
    &=K_R \pi(N-1,K_R-1) \times \\ \nonumber
    &\quad \quad\quad \sum_{i=1}^{K_T} \sum_{\substack{\tau_T \subseteq [K_T] \\ |\tau_T|=i} } \sum_{n=1}^{N} \sum_{j=0}^{K_R} \sum_{\substack{\tau_R \subseteq [K_R] \\ |\tau_R|=j} } (c_{n \tau_T \tau_R}+\tilde{f}_{n \tau_T \tau_R}) \\ \nonumber
    &=K_R \pi(N-1,K_R-1) \times \\ \nonumber
    &\quad \quad\quad \sum_{i=1}^{K_T} \sum_{j=0}^{K_R} \sum_{\substack{\tau_T \subseteq [K_T] \\ |\tau_T|=i} }   \sum_{\substack{\tau_R \subseteq [K_R] \\ |\tau_R|=j} } \sum_{n=1}^{N} (c_{n \tau_T \tau_R}+\tilde{f}_{n \tau_T \tau_R}) \\ \nonumber
    &=K_R \pi(N-1,K_R-1) \sum_{i=1}^{K_T} \sum_{j=0}^{K_R} b_{ij}.
\end{align}

Now we lower bound for the term related to the fronthaul delay as follows:
\begin{align} \label{eq:df}
    &\frac{1}{\pi(N,K_R)} \sum_{\mathbf{d}} {\delta_F^*(\mathbf{d})} \\ \nonumber
    & \stackrel{}{\geq} \frac{1}{\pi(N,K_R)} \sum_{\mathbf{d}}  \frac{1}{K_T} \sum_{i=1}^{K_T} \frac{1}{Fr} \sum_{n \in \mathbf{d}} \sum_{\substack{\tau_T \subseteq [K_T] \\ i \in \tau_T}} f_{n\tau_T} (\mathbf{d}) \\ \nonumber
    &\stackrel{}{=} \frac{1}{\pi(N,K_R)} \frac{1}{K_T Fr}\sum_{\mathbf{d}} \sum_{n \in \mathbf{d}}   \sum_{i=1}^{K_T} i  \sum_{\substack{\tau_T \subseteq [K_T] \\ |\tau_T|=i}} f_{n\tau_T} (\mathbf{d}) \\ \nonumber
    &\stackrel{}{=} \frac{1}{\pi(N,K_R)} \frac{1}{K_T Fr} \sum_{i=1}^{K_T} i \times \\ \nonumber
    &\quad\quad\quad\quad \sum_{\substack{\tau_T \subseteq [K_T] \\ |\tau_T|=i}}      K_R \pi(N-1,K_R-1) \sum_{n=1}^{N} \tilde{f}_{n\tau_T} \\ \nonumber
    &\stackrel{}{=}\frac{K_R}{NK_TFr} \sum_{i=1}^{K_T} i \sum_{\substack{\tau_T \subseteq [K_T] \\ |\tau_T|=i}} \sum_{n=1}^{N} \tilde{f}_{n\tau_T} \\ \nonumber
    &\stackrel{}{=}\frac{K_R}{NK_TFr} \sum_{i=1}^{K_T} i \left(\sum_{j=0}^{K_R} b_{ij} - \sum_{\substack{\tau_T \subseteq [K_T] \\ |\tau_T|=i}} \sum_{n=1}^{N} c_{n\tau_T}  \right) \\ \nonumber
    &\stackrel{}{=}\frac{K_R}{K_Tr}\left(\frac{1}{NF} \sum_{i=1}^{K_T} \sum_{j=0}^{K_R}ib_{ij} - \frac{1}{NF} \sum_{i=1}^{K_T} i \sum_{\substack{\tau_T \subseteq [K_T] \\ |\tau_T|=i}} \sum_{n=1}^{N} c_{n\tau_T}  \right) \\ \nonumber
    &\stackrel{(a)}{\geq} \frac{K_R}{K_Tr}\left(\frac{1}{NF} \sum_{i=1}^{K_T} \sum_{j=0}^{K_R}ib_{ij} - \mu_T K_T\right),
\end{align}
in which we have defined
\begin{align} \nonumber
    c_{n\tau_T} &\triangleq \sum_{j=0}^{K_R} \sum_{\substack{\tau_R \subseteq [K_R] \\ |\tau_R|=j} } c_{n\tau_T\tau_R} \\ \nonumber
    f_{n\tau_T} &\triangleq \sum_{j=0}^{K_R} \sum_{\substack{\tau_R \subseteq [K_R] \\ |\tau_R|=j} } f_{n\tau_T\tau_R} \\ \nonumber
    \tilde{f}_{n\tau_T} &\triangleq \sum_{j=0}^{K_R} \sum_{\substack{\tau_R \subseteq [K_R] \\ |\tau_R|=j} } \tilde{f}_{n\tau_T\tau_R},
\end{align}
and in (a) we have used the following inequality due to the edge nodes cache size constraint in (18c)
\begin{align}
    \mu_T F N K_T &\geq \sum_{i=1}^{K_T} \sum_{n=1}^{N} \sum_{\substack{\tau_T \subseteq [K_T] \\ i \in \tau_T}} c_{n\tau_T} \\ \nonumber
    &= \sum_{n=1}^{N} \sum_{i=1}^{K_T}  \sum_{\substack{\tau_T \subseteq [K_T] \\ i \in \tau_T}} c_{n\tau_T} \\ \nonumber
    &= \sum_{n=1}^{N} \sum_{i=1}^{K_T} i \sum_{\substack{\tau_T \subseteq [K_T] \\ |\tau_T|=i}}  c_{n\tau_T} \\ \nonumber
     &=\sum_{i=1}^{K_T} i   \sum_{\substack{\tau_T \subseteq [K_T] \\ |\tau_T|=i}}  \sum_{n=1}^{N} c_{n\tau_T}. 
\end{align}
Finally, we can add up the terms corresponding to the lower bounds for the delay and edge NDT to arrive at the NDT lower bound 
\begin{align}
    &\frac{K_R}{K_Tr}\left(\frac{1}{NF} \sum_{i=1}^{K_T} \sum_{j=0}^{K_R}ib_{ij} - \mu_T K_T\right) + \\ \nonumber
    &\frac{K_R(1-\mu_R)}{\frac{n_T}{NF}\sum_{i=1}^{K_T} \sum_{j=0}^{K_R}ib_{ij}+\frac{1}{NF}\sum_{i=1}^{K_T} \sum_{j=0}^{K_R}jb_{ij}} \\ \nonumber
    & \geq \frac{K_R}{K_Tr}\left(x - \mu_T K_T\right) + \frac{K_R(1-\mu_R)}{n_T x + \mu_R K_R},
\end{align}
in which we have defined
\begin{align}
    x \triangleq \frac{1}{NF} \sum_{i=1}^{K_T} \sum_{j=0}^{K_R} i b_{ij},
\end{align}
and have used the following inequality due to cache size constraint at the users stated in (18d)
\begin{align}
    &\frac{1}{NF} \sum_{i=1}^{K_T} \sum_{j=0}^{K_R} j b_{ij} \\ \nonumber
    &=\frac{1}{NF} \sum_{j=1}^{K_R} j \sum_{\substack{\tau_R \subseteq [K_R] \\ |\tau_R|=j} } \sum_{n=1}^{N} \sum_{i=1}^{K_T} \sum_{\substack{\tau_T \subseteq [K_T] \\ |\tau_T|=i} } (c_{n\tau_T\tau_R} + \tilde{f}_{n\tau_T\tau_R}) \\ \nonumber
    & \leq \mu_R K_R.
\end{align}

\section{Proof of Lemma \ref{lem:gap}} \label{pro:gap}

To prove the inequality $f(x)\geq 2 \delta_{up}(r,\mu_T,\mu_R)/3$, we first focus on the minimum $f_{min}$ of $f(x)$.

To this end, we first derive the domain of function $f(x)$. From \eqref{Eq_MainOpt_Converse} and \eqref{eq:df}, we can have the inequalities $K_R(x-\mu_T K_T)/(K_T r)\leq \delta_F^*(\mathbf{d}) \leq \delta_{F_{max}}$. Hence, the maximum value of $x$ is given as $x_{max}=\max\{m_{max},\mu_T K_T\}$. We also have the inequality $\sum_{i=1}^{K_T} \sum_{j=0}^{K_R} i b_{ij}\geq NF$ due to the equality $\sum_{i=1}^{K_T} \sum_{j=0}^{K_R} b_{ij}=NF$ in \eqref{eq:bij}. This yields the necessary condition $x\leq 1$. Also because $x\geq 0$ , the minimum value is given as $x_{min}=\max\{\mu_T K_T,1\}$. Hence, $x$ lies in the interval $[x_{min},x_{max}]$. Since function $f(x)$ is convex for $x>0$, and the only stationary point is $x=m_0$, i.e., $f'(m_0)=0$, where we have $m_0=\sqrt{K_T(1-\mu_R)r/n_T}-K_R \mu_R/n_T$. Therefore, the desired minimum $f_{min}$ is given as 
\begin{align} 
f_{min}=\left\{
\begin{array}{ll} 
f(m_0), &  \text{if}~ x_{min}\leq m_0\leq x_{max}\\
       \min\{f(x_{min}), f(x_{max})\}, & \text{otherwise},
\end{array} 
\right.
\end{align}
which can be rewritten as
\begin{align} \label{fmin}
    \! f_{min}\!=\! \left\{
    \begin{array}{ll} 
       \frac{K_R(m^*(r,\mu_R)-\mu_T K_T)}{K_Tr}+\\ 
       \quad \quad \quad \quad \frac{K_R(1-\mu_R)}{m^*(\mu_R,r)n_T+\mu_R K_R}, &:\mu K_T<m^*(r,\mu_R)  \\
     \frac{K_R(1-\mu_R)}{\mu_T K_Tn_T+\mu_R K_R}, &:\mu K_T\geq m^*(\mu_R,r)
 \end{array} 
\right.
\end{align}
where we have defined
\begin{equation} \label{optimalm}
    m^*(r,\mu_R)=
    \begin{cases}
     \max\{m_0,1\} &: r<r_{th}, \\
     m_{max} &: r \geq r_{th}.
    \end{cases}
\end{equation}

As a result, if we can prove the inequality $f_{min}\geq 2 \delta_{up}(r,\mu_T,\mu_R)/3$, then Lemma \ref{lem:gap} holds immediately. Since $f_{min}$ is quite intractable, we turn to choose a simpler function, denoted as $\delta'_{lb}(r,\mu_T,\mu_R)$ that satisfies $\delta'_{lb}(r,\mu_T,\mu_R)\leq \delta_{lb}(r,\mu_T,\mu_R)$, where we have defined $\delta_{lb}(r,\mu_T,\mu_R)=f_{min}$.

The lower bound $\delta'_{lb}(r,\mu_T,\mu_R)$ is given as
\begin{align}  \label{eq:delpub}
&\delta'_{lb}(r,\mu_T,\mu_R)= \notag \\
& K_R(1-\mu_R)\Big(\frac{(i+2-\mu_T K_T) }{(i+1)n_T+\mu_R K_R}+  \frac{(\mu_T K_T-i-1)}{(i+2)n_T+\mu_R K_R}\Big)
\end{align}
for $\mu_T K_T\in [i,i+1)$, with $m(r,\mu_R) \leq i\leq m_{max}-1$; and
\begin{align} \label{twocases}
&\delta'_{lb}(r,\mu_T,\mu_R)= \notag \\
&\begin{cases}
\mbox{\normalsize\(	\frac{K_R(m^*(r,\mu_R)-\mu_T K_T)}{K_T r}+\frac{K_R(1-\mu_R)}{m^*(r,\mu_R) n_T+\mu_R K_R}\)},  \\
\mbox{\normalsize\( \frac{K_R(m(r,\mu_R)-\mu_T K_T)}{K_T r} +\delta'_{lb}(r,\frac{m(r,\mu_R)}{K_T},\mu_R) \)},   
\end{cases}
\end{align}
for $\mu_T K_T \leq m(r,\mu_R)$, and the first expression is for the regime $m^*(r)\in[ (m(r,\mu_R)-0.5),m(r,\mu_R)]$, while the second expression is for the regime $m^*(r)\in[m(r,\mu_R),(m(r,\mu_R)+0.5)]$, where $m^*(r,\mu_R)$ is given in \eqref{optimalm}.

To proceed, we now prove the inequality $\delta_{lb}(r,\mu_T,\mu_R)\geq \delta'_{lb}(r,\mu_T,\mu_R)$.

\begin{proof}
Since $m(r,\mu_R)$ is the nearest integer point of $m_0$ when $r<r_{th}$, we have the inequality $m(r,\mu_R)+1> m_0$, yielding $m(r,\mu_R)+1> m^*(r,\mu_R)$ for this range of $r$. Furthermore, we also have $m(r,\mu_R)=m^*(r,\mu_R)$ when $r\geq r_{th}$. Hence, the inequality $m(r,\mu_R)+1\geq m^*(r,\mu_R)$ holds for any value of $r$. As a result, for any sub-interval $\mu K_T \in [i, i+1)$, with $m(r,\mu_R)+1 \leq i\leq m_{max}-1$, $\delta_{lb}(r,\mu_T,\mu_R)$, i.e., $f_{min}$, is given as 
$K_R(1-\mu_R)/(\mu_T K_Tn_T+\mu_R K_R)$ from \eqref{fmin} since $\mu_T K_T \geq m^*(r,\mu_R)$. By simple comparison between the above $\delta_{lb}(r,\mu_T,\mu_R)$ and $\delta'_{lb}(r,\mu_T,\mu_R)$ in \eqref{eq:delpub}, we have the inequality $\delta_{lb}(r,\mu_T,\mu_R)\geq \delta'_{lb}(r,\mu_T,\mu_R)$.

For the remaining interval $\mu K_T\leq m(r,\mu_R)+1$, we distinguish the following two cases. 

\emph{Case 1: $m^*(r,\mu_R)\in [(m(r,\mu_R)-0.5), m(r,\mu_R) ]$.} Hence, we have the inequality $m^*(r,\mu_R)\leq m(r,\mu_R)$. For interval $\mu_T K_T\in[m(r,\mu_R), m(r,\mu_R)+1]$, the inequality $\delta_{lb}(r,\mu_T,\mu_R)\geq \delta'_{lb}(r,\mu_T,\mu_R)$ holds with the same reason as above. Moveover, by comparison, we have that $\delta'_{lb}(r,\mu_T,\mu_R)$ in \eqref{twocases} is equal to $\delta_{lb}(r,\mu_T,\mu_R)$ in \eqref{fmin} for $\mu K_T\leq m^*(r,\mu_R)$. Instead, for $\mu K_T \in[m^*(r,\mu_R), m(r,\mu_R)]$, since both $\delta'_{lb}(r,\mu_T,\mu_R)$ in \eqref{twocases} and $\delta_{lb}(r,\mu_T,\mu_R)=K_R(1-\mu_R)/(\mu_T K_Tn_T+\mu_R K_R)$ are decreasing functions of $\mu_T$, they are equal for $\mu_T K_T=m^*(r,\mu_R)$, and the former has a smaller gradient for the whole range of value of $\mu_T$ at hand, we have $\delta'_{lb}(r,\mu_T,\mu_R)\leq \delta_{lb}(r,\mu_T,\mu_R)$.

\emph{Case 2: $m^*(r,\mu_R)\in [m(r,\mu_R),(m(r,\mu_R)+0.5)]$.} In this range, the inequality $m^*(r,\mu_R)\geq m(r,\mu_R)$ holds immediately. For interval $\mu_T K_T\in[m(r,\mu_R), m(r,\mu_R)+1]$, since both $\delta'_{lb}(r,\mu_T,\mu_R)$ in \eqref{twocases} and $\delta_{lb}(r,\mu_T,\mu_R)$ are decreasing functions of $\mu_T$, they are equal for $\mu_T K_T=m(r,\mu_R)+1$, and the former has a smaller gradient, we have $\delta'_{lb}(r,\mu_T,\mu_R)\leq \delta_{lb}(r,\mu_T,\mu_R)$. This holds also for the value $\mu_T K_T=m(r,\mu_R)$. Combining this with the fact that $\delta'_{lb}(r, \mu_T,\mu_R)$ in \eqref{twocases} and $\delta_{lb}(r, \mu_T,\mu_R)$ in \eqref{fmin} are linear and parallel for $\mu_T K_T\leq m(r,\mu_R)$, we have $\delta'_{lb}(r, m(r,\mu_R)/K_T,\mu_R) \leq \delta_{lb}(r, \mu_T,\mu_R)$ in this range.
\end{proof}

To complete the proof, we proceed to prove the multiplicative gap between the upper bound $\delta_{up}(r,\mu_T,\mu_R)$ \eqref{Eq_Optimum_NDT} and the lower bound $\delta'_{lb}(r,\mu_T,\mu_R)$. For $\mu_T K_T\in [i,i+1)$, with $m(r,\mu_R) \leq i\leq m_{max}-1$, we have 
\begin{align} \label{final00}
&\frac{\delta_{up}(r,\mu_T,\mu_R)}{\delta'_{lb}(r,\mu_T,\mu_R)} \stackrel{(a)}{\leq} \frac{\delta_{up}(r,\mu_T=i/K_T,\mu_R)}{\delta'_{lb}(r,\mu_T=i/K_T,\mu_R)} \notag \\
&=1+\frac{2n_T}{(in_T+K_R\mu_R)((i+3)n_T+K_R\mu_R)}\leq \frac{3}{2},
\end{align}
where inequality (a) holds because $\delta_{up}(r,\mu_T,\mu_R)$ and $\delta'_{lb}(r,\mu_T,\mu_R)$ are both linearly decreasing and they coincide at the endpoint $\mu_T K_T=i+1$. For $\mu_T K_T\leq m(r,\mu_R)$ in Case 1, the gap is given as 
\begin{subequations} 
\begin{align}
&\frac{\delta_{up}(r,\mu_T,\mu_R)}{\delta'_{lb}(r,\mu_T,\mu_R)}\stackrel{(a)}{\leq} \frac{\delta_{up}(r,\mu_T=m(r,\mu_R)/K_T,\mu_R)}{\delta'_{lb}(r,\mu_T=m(r,\mu_R)/K_T,\mu_R)}  \notag \\
  &=\frac{\frac{1}{m(r,\mu_R)n_T+\mu_R K_R}}{\frac{m^*(r,\mu_R)-m(r,\mu_R)}{K_T r}+\frac{1}{m^*(r,\mu_R)n_T+\mu_R K_R}} \notag  \\
  &\stackrel{(b)}{=}\frac{1/p(m(r,\mu_R)n_T)}{2/p(m^*(r,\mu_R)n_T)-p(m(r,\mu_R)n_T)/p(m^*(r,\mu_R)n_T)} \notag  \\
& \stackrel{(c)}{\leq} 1+\frac{1/4}{(m(r,\mu_R)n_T+\mu_R K_R)^2-(m(r,\mu_R)n_T+\mu_R K_R)} \notag \\
&\stackrel{(d)}{\leq} \frac{3}{2}, \notag 
\end{align}
\end{subequations}
where inequality (a) holds because $\delta_{up}(r,\mu_T,\mu_R)$ and $\delta'_{lb}(r,\mu_T,\mu_R)$ decrease with the same slope and the maximum ratio is at the endpoint $\mu_T K_T=m(r,\mu_R)$; equality (b) holds due to the definition of $p(x)=x+\mu_R K_R$ and $m^*(r,\mu_R)$; inequality (c) holds due to the constraints $m^*(r,\mu_R) \in [( m(r)-0.5), m(r)]$; and inequality (d) holds for any $m(r)\geq 2$, while for $m(r)=1$, we have $K_T r/n_T \in[0,1]$ and $\mu\in[0,1]$. With simple comparison, we can get $\delta_{up}(r,\mu_T,\mu_R)=\delta'_{lb}(r,\mu_T,\mu_R)$. Finally, for $\mu K_T\leq m(r)$ in Case 2, the gap is given as 
\begin{align}
&\frac{\delta_{up}(r,\mu_T,\mu_R)}{\delta'_{lb}(r,\mu_T,\mu_R)}\stackrel{(a)}{\leq} \frac{\delta_{up}(r,\mu_T=m(r,\mu_R)/K_T,\mu_R)}{\delta'_{lb}(r,\mu_T=m(r,\mu_R)/K_T,\mu_R)} \notag \\
&=1+\frac{2n_T}{(m(r,\mu_R)n_T+K_R\mu_R)((m(r,\mu_R)+3)n_T+K_R\mu_R)}
\notag \\
&\leq \frac{3}{2}, 
\end{align}
where inequality (a) holds as inequality (a) in \eqref{final00}. This completes the proof.

\end{appendices}

\end{document}

%% file: defs.tex
\newcommand{\bit}{\begin{itemize}}
\newcommand{\eit}{\end{itemize}}

\newcommand{\bc}{\begin{center}}
\newcommand{\ec}{\end{center}}

\newcommand{\ba}{\begin{array}}
\newcommand{\ea}{\end{array}}

\newcommand{\beq}{\begin{equation}}
\newcommand{\eeq}{\end{equation}}

\newcommand{\beqn}{\begin{equation*}}
\newcommand{\eeqn}{\end{equation*}}

\newcommand{\bean}{\begin{eqnarray*}}
\newcommand{\eean}{\end{eqnarray*}}
\newcommand{\bea}{\begin{eqnarray}}
\newcommand{\eea}{\end{eqnarray}}

